\documentclass[10pt,journal,english,twocolumn]{IEEEtran}
\usepackage{blindtext}

\usepackage{babel}
\usepackage[babel]{microtype}
\usepackage[T1]{fontenc}

\usepackage{graphicx}
\usepackage{xcolor}
\usepackage{tikz}
\usepackage{pgfplots}

\usepackage{tabularx}
\usepackage{booktabs}
\usepackage{stfloats} %

\usepackage{amsmath}
\usepackage{amsfonts}
\usepackage{amssymb}
\usepackage{amsthm}
\usepackage{bm}
\usepackage{siunitx}

\usepackage{csquotes}
\usepackage[backend=biber,url=false,style=ieee,isbn=false,doi=false]{biblatex}
\bibliography{literature.bib}
\renewcommand*{\bibfont}{\footnotesize}

\usepackage{hyperref}
\hypersetup{
	pdfauthor={Karl-Ludwig Besser},
	pdftitle={Copula-Based Bounds for Multi-User Communications -- Part I: Average Performance},
	colorlinks=true,
	linkcolor=black,
	citecolor=black,
	allcolors=.,
	urlcolor=blue,
}

\usepackage[acronyms,nomain,xindy]{glossaries}
\makeglossaries
\newacronym{5g}{5G}{Fifth-Generation Wireless Systems}
\newacronym{ask}{ASK}{amplitude-shift keying}
\newacronym[plural={AWGN}, \glsshortpluralkey={AWGN}]{awgn}{AWGN}{additive white Gaussian noise}

\newacronym{ber}{BER}{bit error rate}
\newacronym{bler}{BLER}{block error rate}
\newacronym{bpsk}{BPSK}{binary phase-shift keying}

\newacronym{cdf}{CDF}{cumulative distribution function}
\newacronym{cnn}{CNN}{convolutional neural network}
\newacronym{csi}{CSI}{channel state information}
\newacronym{csit}{CSI-T}{channel state information at the transmitter}

\newacronym{ecc}{ECC}{error-correcting code}

\newacronym{ff}{FF}{feed-forward}
\newacronym{ffnn}{FF-NN}{feed-forward neural network}
\newacronym{fsk}{FSK}{frequency-shift keying}

\newacronym{gm}{GM}{Gaussian mixture}

\newacronym{ic}{IC}{interference channel}
\newacronym{iid}{i.\,i.\,d.}{independent and identically distributed}
\newacronym{il}{IL}{information leakage}
\newacronym{iot}{IoT}{internet of things}

\newacronym{lb}{LB}{lower bound}
\newacronym{los}{LOS}{line-of-sight}

\newacronym{mac}{MAC}{multiple access channel}
\newacronym{map}{MAP}{maximum a posteriori}
\newacronym{ml}{ML}{machine learning}
\newacronym{mop}{MOP}{multi-objective programming problem}
\newacronym{mse}{MSE}{mean squared error}

\newacronym{nlos}{NLOS}{non-line-of-sight}
\newacronym{nn}{NN}{neural network}
\newacronym{nve}{NVE}{normalized validation error}

\newacronym{pdf}{PDF}{probability density function}
\newacronym{pwc}{PWC}{polar wiretap code}

\newacronym{qam}{QAM}{quadrature amplitude modulation}

\newacronym{rbf}{RBF}{radial basis function}
\newacronym{relu}{ReLU}{rectified linear unit}
\newacronym{ris}{RIS}{reconfigurable intelligent surface}
\newacronym{rnn}{RNN}{recurrent neural network}

\newacronym{sgd}{SGD}{stochastic gradient descent}
\newacronym{sgp}{SGP}{secure goodput}
\newacronym{sinr}{SINR}{signal-to-interference-plus-noise ratio}
\newacronym{snr}{SNR}{signal-to-noise ratio}
\newacronym{svm}{SVM}{support vector machine}

\newacronym{tin}{TIN}{treating interference as noise}
\newacronym{tvd}{TVD}{total variation distance}

\newacronym{uav}{UAV}{unmanned aerial vehicle}
\newacronym{ub}{UB}{upper bound}
\newacronym{urllc}{URLLC}{ultra-reliable low latency communication}
\setacronymstyle{long-short}

\pgfplotsset{compat=newest}
\pgfplotsset{plot coordinates/math parser=false}

\usetikzlibrary{patterns,arrows,positioning,external}
\usepgfplotslibrary{patchplots}

\newcommand{\diff}{\ensuremath{\mathrm{d}}}
\newcommand{\expect}[2][]{\ensuremath{\mathbb{E}_{#1}\left[#2\right]}}
\newcommand{\inv}[1]{\ensuremath{#1^{-1}}}

\newcommand{\positive}[1]{\ensuremath{\left[#1\right]^{+}}}
\newcommand{\leqone}[1]{\ensuremath{\left[#1\right]^{\leq 1}}}

\newcommand{\X}{\ensuremath{\bm{X}}}

\newcommand{\Y}{\ensuremath{\bm{Y}}}

\newcommand{\Z}{\ensuremath{\bm{Z}}}
\newcommand{\xt}{\ensuremath{\tilde{x}}}
\newcommand{\yt}{\ensuremath{\tilde{y}}}
\newcommand{\lx}{\ensuremath{{\lambda_{x}}}}
\newcommand{\ly}{\ensuremath{{\lambda_{y}}}}

\newcommand{\ratemac}{\ensuremath{R_{\text{MAC}}}}

\renewcommand{\Pr}{\ensuremath{\mathbb{P}}}

\theoremstyle{plain}%
\newtheorem{thm}{Theorem}%
\newtheorem{lem}[thm]{Lemma}

\theoremstyle{definition}
\newtheorem{defn}{Definition}%

\theoremstyle{remark}
\newtheorem{rem}{Remark}
\newtheorem*{rem*}{Remark}
\newtheorem{example}{Example}

\definecolor{plot2}{HTML}{127933}
\definecolor{plot1}{HTML}{242bb3}
\definecolor{plot0}{HTML}{F04D4A}
\definecolor{plot3}{HTML}{cc9f2b}
\definecolor{plot4}{HTML}{27AFB3}

\definecolor{change}{HTML}{0096b8}%

\title{Copula-Based Bounds for Multi-User Communications -- Part I: Average Performance}
\author{Eduard A. Jorswieck, \IEEEmembership{Fellow, IEEE} and Karl-Ludwig Besser, \IEEEmembership{Student Member, IEEE}
	\thanks{The authors are with the Institute of Communications Technology, Technische Universit\"at Braunschweig, 38106 Braunschweig, Germany (email: \{{e.jorswieck}, {k.besser}\}@tu-bs.de).}
	\thanks{This work is supported in part by the German Research Foundation (DFG) under grant JO\,801/23-1.}
}
\IEEEspecialpapernotice{(Invited Paper)}

\begin{document}
\maketitle

\begin{abstract}%
	Statistically independent or positively correlated fading models are usually applied to compute the average performance of wireless communications. However, there exist scenarios with negative dependency and it is therefore of interest how different performance metrics behave for different general dependency structures of the channels. Especially best-case and worst-case bounds are practically relevant as a system design guideline.
	In this two-part letter, we present methods and tools from dependency modeling which can be applied to analyze and design multi-user communications systems exploiting and creating dependencies of the effective fading channels. The first part focuses on fast fading with average performance metrics, while the second part considers slow fading with outage performance metrics.
\end{abstract}

\begin{IEEEkeywords}
Network reliability, Joint distributions, Fading channels, Ergodic performance, Fast fading.
\end{IEEEkeywords}

\section{Introduction and Motivation}
In modern communication systems, multiple links are established at the same time. On the one hand, this is done to serve multiple users at the same time. On the other hand, using multiple antennas exploits spatial diversity~\cite{Heath2018}.
A conventional assumption in the literature is that all of these channels are statistically independent~\cite{Simon2000} or experience positively correlated fading as in the Kronecker model~\cite{Wang2007}. However, real measurements, e.g., for spectral diversity systems, show that this assumption does not always hold in practice~\cite{Lee1973,Peters2014}. The theoretical side of this has not been studied extensively. A relatively recent work \cite{Biglieri2016} shows the impact that the (potentially negative) dependence of channels can have on the performance of wireless communication systems.

In this first part of the two-part letter, we focus on the ergodic performance in the context of fast-fading channels. First, we will give a brief example for arbitrarily correlated channels in Section~\ref{sec:motivation-example}. Next, we introduce theoretical foundations and tools for dependency analysis in Section~\ref{sec:expectation-bounds-theory}. The presented methods will be illustrated with some plastic examples. Applications to problems in the context of communications are provided in Section~\ref{sec:expectation-bounds-applications}. Methods and tools for the analysis of the outage probability for scenarios with slow-fading channels can be found in the second part~\cite{Besser2020part2}.

Multi-user communications over fading channels is a challenging task of constant interest over the last decades \cite{Biglieri1998}. Typical wireless channels have random fading channel coefficients due to multipath propagation  \cite{Rappaport2002}. While it is well established to perform wireless channel measurements of a single link \cite{Bello1963} between one transmitter and one receiver and to develop stochastic channel models based on different scenarios and parameters \cite{Fleury1999}, it is much more involved to measure wireless multi-point channels and to derive corresponding stochastic channel models \cite{Xu2000}. 

Depending on the fading channel characteristics, the operating regimes slow and fast fading, as well as the corresponding performance metrics average (or ergodic) and outage capacity are distinguished \cite{Tse2005}. Indeed, there are some examples of multiuser channels where the fundamental limits do not depend on the joint distribution of the channels: these are all channels where the same marginal property for the capacity region holds \cite{Moser2019}, e.g., broadcast channels. There, the capacity region depends on the marginal conditional probabilities of the received signal, given the channel input, but not on the joint distribution of the received signals \cite{Cover1972}. 

In most multi-user communication scenarios, the performance depends on the joint distribution of the fading channel realizations. Most notably, this can be observed in receive diversity schemes, such as antenna arrays or frequency diversity schemes \cite{Hamdi2008,Zhu2019}. However, the impact of the dependency of random variables on entropy in general \cite{Cover1994} and of the independence assumption in wireless communication analysis \cite{Biglieri2016} has gained attention. In most of the previous work, only linear correlation is considered to describe the dependency. Copulas~\cite{Nelsen2006} on the other hand allow modeling general dependency structures and have also already been used in the area of communications. In \cite{Peters2014}, it was shown that real channel measurements can follow a nonlinear dependency structure which can be modeled using copulas. In \cite{Ghadi2020}, the outage probability for Rayleigh fading channels following a certain dependency structure is derived. Copulas have also been used to model interference in \gls{iot} wireless networks~\cite{Zheng2019}.

In \cite{Haber1974}, a method is proposed how negatively correlated channels can be constructed. As noted in \cite{Akki1985}, this could have a relevant application in air-to-ground communication with low flying aircrafts. Recently, wireless communication networks including \glspl{uav} have gained attention~\cite{Cao2018,Mozaffari2019}. It might therefore be possible to design transmission strategies which actively control the dependency structure of the channels in such scenarios. The bounds provided in this work can then be used as performance benchmarks and design guidelines.

Finally, in \cite{Lin2019} the freedom to choose a particular joint distribution for fixed marginal fading distributions is exploited to derive novel ergodic capacity region results for broadcast and classes of interference channels.

\textit{Notation:}
Throughout this work, we use the following notation. Random variables are denoted in capital boldface letters, e.g., $\X$, and their realizations in small letters, e.g., $x$. We will use $F$ and $f$ for a probability distribution and its density, respectively. The expectation and variance are denoted by $\mathbb{E}$ and $\mathbb{V}$, respectively, and the probability of an event by $\Pr$. It is assumed that all considered distributions are continuous. The uniform distribution on the interval $[a,b]$ is denoted as $\mathcal{U}[a,b]$.
As a shorthand, we use $\positive{x}=\max\left[x, 0\right]$; and similarly $\leqone{x}=\min\left[x, 1\right]$.
The derivative of a univariate function $g$ is written as $g^{\prime}$.
The real numbers and extended real numbers are denoted by $\mathbb{R}$ and $\bar{\mathbb{R}}$, respectively. Logarithms, if not stated otherwise, are assumed to be with respect to the natural base.
\section{Example for Simple Correlated Channels}\label{sec:motivation-example}
First, we provide a brief example where fading channels can have an arbitrary (also negative) correlation.
It is a simplified example for illustration purposes only.

The setup is the following.
We consider an (infinitely) large plain ground. A transmitter with a single antenna is placed at height $h_{\text{Tx}}$. Two receive antennas are placed on top of each other at heights $h_1$ and $h_2=h_1+\Delta h$. The distance (on the ground) between them and the transmitter is $d$.
The transmitter transmits a signal $x(t)$ to the two receive antennas. Each antenna $i$ receives a \gls{los} signal after a delay of $\tau_{0,i}$. Due to the reflecting ground, they also receive a \gls{nlos} signal. This component arrives at the receivers after the delays $\tau_{1,i}, i=1,2$.
This scenario could occur when an \gls{uav} flies above a calm water surface or in a flat rural environment \cite{Akki1985}. 

A similar example for a simple two path fading model is given in \cite{Haber1974}. The difference to \cite{Haber1974} is that they consider a system with only one receive antenna but multiple transmission frequencies. However, we can re-use the signal model from \cite{Haber1974} and extend it to our scenario in the following.
The received signal at antenna $i$ is given as~\cite[Eq.~(1)]{Haber1974}
\begin{equation}
r_i(t) = A_1 \cos\left(\omega (t-\tau_{0,i})\right) + A_{2}\cos\left(\omega(t-\tau_{1,i})\right)\,.
\end{equation}
The squared envelope is then given by~\cite[Eq.~(2)]{Haber1974}
\begin{equation}\label{eq:def-square-envelope}
X_i = A_1^2 + A_2^2 + 2 A_1 A_2 \cos\left(\omega(\tau_{0,i}-\tau_{1,i})\right)\,.
\end{equation}
The delays $\tau$ are calculated using the relation $\tau=s/c$, where $c$ is the propagation speed and $s$ is the respective distance. The distances $s$ (\gls{los} and \gls{nlos}) are calculated using basic trigonometry. The detailed calculations and simulations can be found at \cite{BesserGitlab}.

In Fig.~\ref{fig:motivation-rec-envelopes}, the received envelopes $X_i$ of receive antennas 1 and 2 are shown over the distance $d$ for different values of $\Delta h$.
It can be seen that $X_1$ and $X_2$ are positively correlated for $\Delta h=\SI{0.05}{\meter}$ with correlation coefficient $\rho=0.31$, while they are negatively correlated for $\Delta h=\SI{0.1}{\meter}$ with $\rho=-0.64$ (for a uniform distribution of $d$ between \SIlist{20;50}{\meter}).
\begin{figure}
	\centering
	\begin{tikzpicture}%
\begin{axis}[
	width=\linewidth,
	height=0.23\textheight,
	xmin=20, xmax=50,
	ymin=0, ymax=3,
	xlabel={Distance $d$},
	ylabel={Receive Envelopes $X_i$},
	legend cell align=left,
]

\addplot[plot0, mark=*, mark repeat=30, smooth, thick] table[x=distance, y=x1] {data/rec_evelopes-f2.00E+09-htx10-h11.0-dh0.05.dat};
\addlegendentry{$X_1$};

\addplot[plot1, mark=square*, mark repeat=30, smooth, dashed] table[x=distance, y=x2] {data/rec_evelopes-f2.00E+09-htx10-h11.0-dh0.05.dat};
\addlegendentry{$X_2$ with $\Delta h=\SI{0.05}{\meter}$};

\addplot[plot2, mark=triangle*, mark repeat=30, smooth] table[x=distance, y=x2] {data/rec_evelopes-f2.00E+09-htx10-h11.0-dh0.1.dat};
\addlegendentry{$X_2$ with $\Delta h=\SI{0.1}{\meter}$};
\end{axis}
\end{tikzpicture}
	
	\vspace*{-1em}
	\caption{Received envelopes $X_i$ over distance $d$ for different values of $\Delta h$. The parameters are $A_1=1$, $A_2=0.5$, $f=\SI{2}{\giga\hertz}$, $h_{\text{Tx}}=\SI{10}{\meter}$, and $h_1=\SI{1}{\meter}$.}
	\label{fig:motivation-rec-envelopes}
\end{figure}
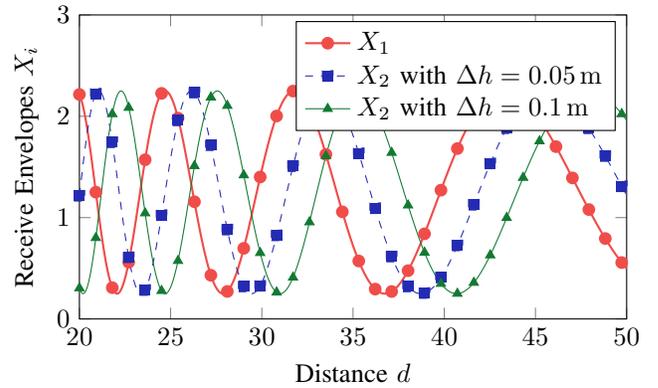
In the interactive supplementary material~\cite{BesserGitlab}, we also give a simple geometry-based simulation, where multiple receivers with fixed antenna heights are placed at random positions around the transmitter and the correlation is estimated. We encourage the reader to try different parameter constellations and observe the behavior.
\section{Bounds on the Expected Value}\label{sec:expectation-bounds-theory}
One performance measure, which is of interest to a system designer, is the ergodic throughput. In the case of fast fading, the channel varies during the transmission. Therefore, the average value of metrics like the ergodic channel capacity is used to evaluate the system performance~\cite{Tse2005}. {Weighted average mean-square error expressions were also applied to quantify the average system performance \cite{Qin2017}.}
In this section, we will present results which allow bounding the expected value of a function of random variables. The results originate from optimal mass transport and are taken from \cite{Rachev1998}.

First, we will state the problem formulation. 
We are given two random variables $\X$ and $\Y$ over the real numbers with fixed marginals $F_{\X}$ and $F_{\Y}$, e.g., measured fading distributions at different locations. We are now interested in the best upper and lower bounds on the expected value of a performance measure $c(\X, \Y)$, e.g., the sum rate, over all possible joint distributions of $\X$ and $\Y$. Mathematically speaking, we want to find %
\begin{equation*}
\inf_{F_{\X, \Y}} \expect[(\X, \Y)]{c(\X, \Y)} \enspace \text{and} \enspace \sup_{F_{\X, \Y}} \expect[(\X, \Y)]{c(\X, \Y)}%
\end{equation*}
{for fixed marginal distributions $F_{\X}$ and $F_{\Y}$.}

For specific functions $c$, the solutions for the upper and lower bound are attained for comonotonic and countermonotonic random variables, respectively. An exact definition based on copulas will be given in the second part.
The central sufficient condition, which the cost $c$ can fulfill, is the Monge condition~\cite[(3.1.7)]{Rachev1998}.
\begin{defn}[{Monge Condition \cite[(3.1.7)]{Rachev1998}}]\label{def:monge-condition}
	A cost function $c: \mathbb{R}^2\to\mathbb{R}$ which satisfies the Monge condition is right-continuous and fulfills
	\begin{equation}\label{eq:monge-condition}
	c(x', y') + c(x, y) \leq c(x, y') + c(x', y) \,,
	\end{equation}
	for all $x'\geq x$ and $y'\geq y$.
\end{defn}
With this definition, we are able to restate the result from \cite[Thm.~3.1.2(b)]{Rachev1998} about the bounds on the expected value of $c(\X, \Y)$.
\begin{thm}[{\cite[Thm.~3.1.2(b)]{Rachev1998}}]\label{thm:bounds-ruschendorf}
	Let $F_{\X,\Y}$ be a distribution function on $\mathbb{R}^2$ with marginals $F_{\X}$, $F_{\Y}$ and let $(\bm{X}, \bm{Y})\sim F_{\X, \Y}$. Suppose that $c$ satisfies the Monge condition and that $\expect[({\X, \Y})]{c(\bm{X}, \bm{Y})}$ exists and is finite. Then
	\begin{align}
	\int_{0}^{1} c(\inv{F_{\bm{X}}}(u), \inv{F_{\bm{Y}}}(u))\diff{u} &\leq \expect{c(\bm{X}, \bm{Y})} \label{eq:thm-bounds-ruschendorf-lower}\\
	\int_{0}^{1} c(\inv{F_{\bm{X}}}(u), \inv{F_{\bm{Y}}}(1-u))\diff{u} &\geq \expect{c(\bm{X}, \bm{Y})} \label{eq:thm-bounds-ruschendorf-upper}
	\end{align}
	holds.
\end{thm}

{
\begin{rem}\label{rem:topkis-theorem}
	A function $c$ that fulfills the Monge condition is also called \emph{submodular}. The function $-c$ is then called supermodular~\cite{Puccetti2015}. If $c$ is twice continuously differentiable, an equivalent definition of the Monge condition is given by Topkis's characterization theorem~\cite{Milgrom1990}. It states that a function $c$ is submodular, i.e., it fulfills the Monge condition, if $\partial^2 c(x,y)/\partial x\partial y\leq 0$ holds for all $x, y$.
\end{rem}
}

\begin{example}
	As a first example, we will take a look at the {\gls{sinr} as the} cost function $c(x, y) = \frac{x}{1+y}$, 
	with exponentially distributed random variables, i.e., $\X\sim\exp(\lx)$ and $\Y\sim\exp(\ly)${\footnote{Any other marginal distribution, e.g., Nakagami-$m$ or even heterogeneous distributions, e.g. $\X$ Ricean and $\Y$ log-normal, works here, too.}}. {The signal of interest in this case is $\X$, while $\Y$ represents the interference.}
	First, we need to show that $c$ satisfies the Monge condition~\eqref{eq:monge-condition}. This can be done as follows
	\begin{align*}
	c(x', y') + c(x, y) - c(x, y') - c(x', y)&\leq 0\\
	\Leftrightarrow x'(1+y) + x(1+y') - x'(1+y') - x(1+y) &\leq 0\\
	\Leftrightarrow (x'-x)(y-y') \leq 0\,,
	\end{align*}
	with $x'\geq x$ and $y'\geq y$ by definition.
	Therefore, we can apply Theorem~\ref{thm:bounds-ruschendorf} to bound the expected value.
	For exponentially distributed random variables with $\lx=1$ and $\ly=2$, this can be evaluated to $0.555 \leq \expect[(\X, \Y)]{\frac{\X}{1+\Y}}\leq 0.870$.
	For comparison, the expected value for the case of independent $\X$ and $\Y$ is around \num{0.723}.
	The detailed calculations can be found at \cite{BesserGitlab}.
\end{example}

{It can be seen from Theorem~\ref{thm:bounds-ruschendorf}, that the bounds are attained for comonotonic and countermonotonic random variables $\X$ and $\Y$. This means that both bounds are individually tight and are achieved by different joint distributions. All other possible joint distributions of $\X$ and $\Y$ achieve average performances between the bounds from Theorem~\ref{thm:bounds-ruschendorf}.} %

{
\begin{rem}\label{rem:monge-condition-minus-c}
	Note that it is also possible to apply Theorem~\ref{thm:bounds-ruschendorf}, if $-c$ fulfills the Monge condition.
	In this case, the following holds
	\begin{equation}
	-\mathsf{UB}_{-c} \leq \expect{c(\X, \Y)} \leq -\mathsf{LB}_{-c}\,,\\
	\end{equation}
	where $\mathsf{UB}_{-c}$ and $\mathsf{LB}_{-c}$ are the upper and lower bound on the expected value of $-c$ according to \eqref{eq:thm-bounds-ruschendorf-upper} and \eqref{eq:thm-bounds-ruschendorf-lower}, respectively.
\end{rem}
}

{
\begin{rem}
	It is possible to extend the results to more than two random variables numerically.
	\cite{Puccetti2015a} offers an algorithm to numerically compute sharp bounds on the expected value of supermodular functions with fixed marginals. This can also be applied to measured channel data.
\end{rem}
}
\section{Bounds on the Ergodic Performance}\label{sec:expectation-bounds-applications}

{In this section, we will give various examples from the area of communications where the results from the previous section can be applied.}
In the following, we assume to have perfect \gls{csi} at the receiver and statistical \gls{csi} at the transmitter.

\subsection{Ergodic Capacity MAC}
Our first example, is the \gls{mac}. The ergodic capacity region for fast fading channels is derived in \cite[Section 23.5]{Gamal2011}. It is achieved with successive interference cancellation. The achievable rates for fixed decoding order can be expressed as 
\begin{equation}
\begin{split}
\ratemac^{(1)} & = \expect{\log_2\left(1+\frac{\bm{X}}{s+\bm{Y}}\right)}\\  \ratemac^{(2)} &= \expect{\log_2 \left( 1 + \frac{\bm{Y}}{s} \right)}
\end{split}
	 \label{eq:macrates}, 
\end{equation}
with $\bm{X}$ as received power of the first decoded user and $\bm{Y}$ as received power of the second decoded user. $s$ is the noise variance. Obviously, $\ratemac^{(2)}$ only depends on the marginal distribution of $\bm{Y}$. However, $\ratemac^{(1)}$ depends on the joint distribution of $\bm{X}$ and $\bm{Y}$, i.e., on the received powers of both users. 

\begin{lem}
	The achievable instantaneous rate of the first decoded users $\log(1 + x/(s+y))$ fulfills the Monge condition.
\end{lem}
\begin{proof}
{From Remark~\ref{rem:topkis-theorem}, we know that $c$ fulfills the Monge condition, if $\partial^2 c(x,y)/\partial x\partial y\leq 0$ holds.

The needed derivative is given as
\begin{equation}
	\frac{\partial^2 c(x,y)}{\partial x\partial y} = \frac{-1}{(s + x + y)^2}\,,
\end{equation}
which is always less than zero.}%
\end{proof}

As a result, the upper and lower bounds from Theorem~\ref{thm:bounds-ruschendorf} on the average rate $\ratemac^{(1)}$ apply. For Rayleigh fading channels, the bounds on the ergodic rate of the first decoded user are illustrated and compared to statistical independent channels in Fig.~\ref{fig:outage}. The interactive source code for different parameter scenarios is provided in \cite{BesserGitlab}.

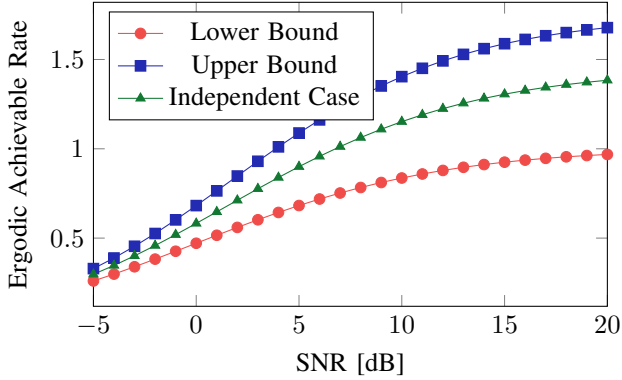
\begin{figure}
	\centering
	\begin{tikzpicture}%
\begin{axis}[
	width=.95\linewidth,
	height=.23\textheight,
	xmin=-5,
	xmax=20,
	xlabel={SNR [dB]},
	ylabel={Ergodic Achievable Rate},
	legend entries={{Lower Bound}, {Upper Bound}, {Independent Case}},
	legend pos=north west,
]
\addplot[plot0, mark=*] table[x=snr,y=min] {data/expectation-mac-rate-lx1-ly1.dat};
\addplot[plot1, mark=square*] table[x=snr,y=max] {data/expectation-mac-rate-lx1-ly1.dat};
\addplot[plot2, mark=triangle*] table[x=snr,y=ind] {data/expectation-mac-rate-lx1-ly1.dat};
\end{axis}
\end{tikzpicture}
	
	\vspace*{-1em}
	\caption{Ergodic capacity of first decoded user in two-user MAC with Rayleigh fading. {The \gls{snr} is given as $\mathsf{SNR}=1/s$.}}
	\label{fig:outage}
\end{figure}

{
\begin{rem}
	Closely related to the achievable rate of the first decoded user in the \gls{mac}, the ergodic sum capacity $\expect{\log_2\left(1+\X+\Y\right)}$ itself depends on the joint distribution of the underlying fading channels. It is easy to verify that the sum rate $\log_2\left(1+x+y\right)$ also fulfills the Monge condition.
\end{rem}
}

\subsection{Ergodic Secret-Key Capacity}
Another example where Theorem~\ref{thm:bounds-ruschendorf} can be applied is the ergodic secret-key capacity. {This average performance metric describes the rate of a secret key generated from using a communication channel and public feedback channel (secret key generation channel model) as described in \cite[Chap.~4]{Bloch2011},}
\begin{equation}
C_{\text{SK}} = \expect{\log_2\left(\frac{1+\X+\Y}{1+\Y}\right)}\,.
\end{equation}
The proof and more details can be found in \cite{Besser2020wsa}. Next to Rayleigh fading, the authors consider the more general case of $\alpha$-$\mu$ fading and show some interesting behavior of the bounds.

\subsection{Proportionally Fair Scheduling}
{After demonstrating a few examples where Theorem~\ref{thm:bounds-ruschendorf} can be directly applied in the context of wireless communications, we want to highlight that it is also possible to apply the theorem, if $-c$ fulfills the Monge condition.}

Consider the ergodic proportional fair rate \cite{bonald2004}, which is the product of the individual rates
\begin{equation}
\expect{\log(1+\X)\log(1+\Y)}\,. \label{eq:pfs}
\end{equation}

The following steps show that this function does not satisfy the Monge condition.
Using
\begin{equation*}
	c(x, y) = \log(\xt)\log(\yt)
\end{equation*}
with $\xt=1+x \leq 1+x'=\xt'$, we get
\begin{align*}
0&\geq \log(\xt')\left(\log(\yt')-\log(\yt)\right) - \log(\xt)\left(\log(\yt')-\log(\yt)\right)\\
0&\geq \underbrace{\left(\log(\xt')-\log(\xt)\right)}_{\geq 0} \underbrace{\left(\log(\yt')-\log(\yt)\right)}_{\geq 0}\,,
\end{align*}
which is not true. {However, it is easy to see that it is true for $-c$. We can therefore calculate the bounds on the expected value of $c$ as described in Remark~\ref{rem:monge-condition-minus-c}.}

{
\subsection{Two-user Collision Channel}
Following the idealized model from \cite[Section II]{Altman2008}, consider a two-user collision channel with two available resource blocks. Both users have a fixed access probability for each channel, i.e., $1 \geq p_i \geq 0$ is the probability that mobile $i$ is active on channel one $\mathbb{P}(M_i = 0) = p_i$. This implies that the probability for mobile $i$ active on channel two is $1-p_i = \mathbb{P}(M_i = 1)$. For statistical independent access probabilities, the resulting transmission success probability is $U(p_1,p_2) =  p_1 (1-p_2) + (1-p_1) p_2$.
If users are allowed to coordinate their access probabilities, the random variables $M_1, M_2$ get dependent with the following joint probability distribution $\mathbb{P}(M_1 = i, M_2 = j) = p_{ij}$. The marginals are $\mathbb{E} M_1 = (1-p_1) = p_{10} + p_{11}$ and $\mathbb{E} M_2 = (1-p_2) = p_{01} + p_{11}$. The correlation coefficient between $M_1, M_2$ is given by
\begin{equation}
\hspace*{-.67em}\rho = \frac{ \mathbb{E}[ M_1 M_2]-\mathbb{E}M_1 \mathbb{E}M_2}{\sqrt{\mathbb{V}(M_1) \mathbb{V}(M_2)}} = \frac{ p_{11} - (1-p_1)(1-p_2)}{\sqrt{ p_1 (1-p_1) p_2 (1-p_2)}} \label{eq:rho}.
\end{equation}
For arbitrarily dependent $M_1, M_2$, the transmission success probability is given by $U(\mathbf{p}) = p_{01} + p_{10}$.
For fixed marginal distributions $p_1, p_2$, the transmission success probability $U = p_{01} + p_{10} = 2- p_1 - p_2 - 2p_{11}$ is equivalent to $-p_{11}$. Since $\rho$ is increasing in $p_{11}$, we conclude that $\rho$ increases iff $U$ decreases, i.e., the success probability is a decreasing function in $\rho$. Please note, that the minimum and maximum correlation coefficient is not always $-1$ and $1$, but it depends on the marginal distributions \cites[Example 16]{Leonov2020}{Embrechts2009}. For $p_1 = p_2 = 1/2$, we obtain minimum and maximum correlation coefficient $-1 \leq \rho \leq 1$ with corresponding success probability between $0 \leq U \leq 1$. 
}
\section{Conclusion and Outlook}\label{sec:conclusion}
The overall goal of the two-part letter is to explain the basics of the methods necessary to study the impact of dependency between random variables with applications to wireless communications. In particular, two typical performance metrics for fast and slow fading channels, the ergodic and the outage performance, depend significantly on the joint distribution of the underlying random fading parameters.
{This first part dealt with the ergodic performance, which is typically used for scenarios with fast-fading channels. In the second part~\cite{Besser2020part2}, we will investigate the outage probability of slow-fading channels.}

The presented tools allow engineers to compute lower and upper bounds for arbitrary dependency structures. The worst-case bounds are then applied for robust system design under uncertainties, while the upper bounds are the best possible achievable dependencies.
{As shown in \cite{Besser2019}, it is possible to have a positive zero-outage capacity for dependent fading channels. It is therefore of interest for future research how such dependency structures can be set up in real communication systems, especially in the context of \gls{urllc}.}
One way to actively achieve them might be by tweaking the radio propagation conditions, e.g., with new emerging technologies, including \glspl{ris} \cite{DiRenzo2019} and the traditional relaying.

{Besides providing performance bounds, copulas can also be used to flexibly model the joint distribution. One can then calculate the system's performance for this specific joint distribution. There exist different parametric copula families that provide a more flexible and general way of modeling the dependency than only considering linear correlation~\cite{Nelsen2006}. One advantage is that they allow modeling tail-dependency. Currently, this is often used in the finance, e.g., for portfolio analysis~\cite{Fischer2015}. In communications, this could be interesting in the context of \gls{urllc}~\cite{Bennis2018} and other applications in 6G~\cite{Dang2020}.}

{Statistical dependency has an important impact in many more applications. These include cognitive radio~\cite{Wang2019}, queuing~\cite{Sun2018}, and improper signaling~\cite{Hellings2015}. A more general dependency analysis using copulas might also be beneficial in the context of finite-length information theory.}

\printbibliography

\end{document}